\documentclass[12pt]{article}

\usepackage{graphicx}
\usepackage{endfloat}
\usepackage{amssymb}
\usepackage{natbib}
\usepackage{fullpage}

\usepackage{amsfonts}
\usepackage{amsthm}
\usepackage{amsmath}

\newtheorem{prp}{Proposition}
\newtheorem{asm}{Assumption}
\newtheorem{lem}{Lemma}

\begin{document}

\title{Structured, Sparse Aggregation}
\author{Daniel Percival  \\
Department of Statistics \\ Carnegie Mellon University \\ Pittsburgh, PA 15213 \\ 
email: \texttt{dperciva@andrew.cmu.edu} }

\maketitle

\newpage

\mbox{}
\vspace*{2in}
\begin{center}
\textbf{Author's Footnote:}
\end{center}
Daniel Percival is Doctoral Candidate in Statistics at Carnegie Mellon University. Mailing address: Department of Statistics, Carnegie Mellon University, Pittsburgh, PA 15213
(email: dperciva@andrew.cmu.edu).  This work was funded by the National Institutes of Health grant
MH057881 and National Science Foundation grant DMS-0943577.  The author would like the thank Larry Wasserman and Kathryn Roeder for helpful comments and discussions.

\newpage
\begin{center}
\textbf{Abstract}
\end{center}
We introduce a method for aggregating many least
squares estimator so that the resulting estimate has two properties:
sparsity and structure. That is, only a few candidate covariates are used
in the resulting model, and the selected covariates follow some
structure over the candidate covariates that is assumed to be known a priori. While sparsity is well studied in many
settings, including aggregation, structured sparse methods are still
emerging. We demonstrate a general framework for structured sparse
aggregation that allows for a wide variety of structures, including
overlapping grouped structures and general structural penalties defined as set
functions on the set of covariates. We show that such estimators
satisfy structured sparse oracle inequalities --- their finite sample risk adapts to the structured sparsity of the target. These inequalities reveal that under suitable settings, the structured sparse estimator performs at least as well as, and potentially much better than, a sparse aggregation estimator. We empirically establish the
effectiveness of the method using simulation and an application to
HIV drug resistance.  

\vspace*{.3in}


\newpage

\section{Introduction}

In statistical learning, sparsity and variable selection are well
studied and fundamental topics. Given a large set of candidate
covariates, sparse models use only a few in the
model. Sparse techniques often improve out of sample performance and
aid in model interpretation. We focus on the linear regression setting. Here, we
model a vector of responses $\textbf{y}$ as a linear combination of
$M$ predictors, represented as an $n \times M$ data matrix $\textbf{X}$, via
the equation $\textbf{y} = \textbf{X}\boldsymbol{\beta} + \boldsymbol{\epsilon}$, where $\boldsymbol{\beta}$ is a vector
of linear coefficients and $\boldsymbol{\epsilon}$ is a vector of stochastic noise.  The task is then to produce an 
estimate of $\boldsymbol{\beta}$, denoted $\widehat{\boldsymbol{\beta}}$, using $\textbf{X}$ and $\textbf{y}$.
Sparse modeling techniques produce a $\widehat{\boldsymbol{\beta}}$ with only a few
nonzero entries, with the remaining set equal to zero, effectively
excludes many covariates from the model.  One example of a sparse regression method is the lasso estimator~\citep{tibshiraniLasso}:
\begin{align}
\widehat{\boldsymbol{\beta}}_{\mbox{lasso}} = \underset{\boldsymbol{\beta} \in
  \mathbb{R}^M}{\operatorname{argmin}} \|\textbf{y} -
\textbf{X}\boldsymbol{\beta}\|_2^2 + \lambda \sum_{j=1}^M |{\beta}_j|.
\end{align}
In the above, $\lambda>0$ is a tuning parameter.  Here, the $\ell_1$
penalty encourages many entries of
$\widehat{\boldsymbol{\beta}}_{\mbox{lasso}}$ to be identically zero, giving a
sparse estimator.

Suppose now that additional structural information is available
about the covariates. We then seek to incorporate this information in
our sparse modeling strategy, giving a {\em structured, sparse}
model. For example, consider a factor covariate with
$u$ levels, such as in an ANOVA model, encoded
as a set of $u-1$ indicator variables in $\textbf{X}$.  Taking this
 structure into account, we then would jointly
select or exclude this set of covariates from our sparse model. More
generally, suppose that we have a graph with $M$ nodes, each node
corresponding to a covariate. This graph might represent a spatial relationship between the covariates. A sparse model incorporating this information might jointly include or exclude sets of predictors corresponding to neighborhoods or
cliques of the graph. In summary, sparsity seeks a $\widehat{\boldsymbol{\beta}}$ with few nonzero entries, whereas structured sparsity seeks a sparse $\widehat{\boldsymbol{\beta}}$ where the nonzero entries following some a priori defined pattern.

As an example, consider the results displayed in Figure~\ref{twoDimOneRunFigure}. In
the top left, we see a coefficient vector, rearranged as a square matrix. The nonzero entries, represented as white squares, have a clear structure with respect to the familiar two dimensional lattice. On the bottom row, we display
the results of two sparse methods, including the lasso. The top right
panel displays the results of one of the methods of this paper. Since our
method also takes the structural
information into account, it is able to more
accurately re-create the sparsity pattern pictured in the top left.

Though methods for structured sparsity are still emerging, there are now many
examples in the literature. The grouped
lasso~\citep{Yuan06modelselection} allows for joint selection of
covariates, where the groups of covariates partition the set of
covariates. Subsequent work
\citep*{Huang:2009:LSS:1553374.1553429,Jacob:2009:GLO:1553374.1553431,Jenatton09structuredsparse}
extended this idea to allow for more flexible structures based on overlapping groups of covariates. 
 Further, \cite{Bach10hierarchical} and~\cite*{Zhao-cap-penalty} proposed
 methods for Hierarchical structures
 and~\cite{xingTreeGuided} as well as \cite{pengMasterPredictors} gave methods
 in the multi-task setting for coherent variable selection across tasks.


In this paper, we present an aggregation estimator that produces
structured sparse models.  In the linear regression setting,
aggregation estimators combine many estimates of $\boldsymbol{\beta}$:
$\{\widehat{\boldsymbol{\beta}}_1,\ldots,\widehat{\boldsymbol{\beta}}_B\}$ in some way to give
an improved estimate $\widehat{\boldsymbol{\beta}}_{\mbox{Aggregate}}$.  See
\cite*{buneaAggregation} and the references therein
 for discussions of aggregation in general
settings, and~\cite{Yang_adaptiveregression,Yang_adaptiveregressionBern} for methods in the linear regression setting. In particular, we extend the methods and results given
by~\cite{RigTsy10}, who focused on sparse aggregation, where the
estimated $\widehat{\boldsymbol{\beta}}_{\mbox{Aggregate}}$ has many entries equal
to zero. Their sparse aggregation method combines in a weighted
average the least squares
estimates for each subset of the
set of candidate covariates.
For a particular model in the average, its weight
is, in part, inversely exponentially proportional to the number of covariates in the
model, i.e. the sparsity of the model. This strategy
encourages a sparse $\widehat{\boldsymbol{\beta}}_{\mbox{Aggregate}}$. We extend this idea by
proposing an alternate set of weights that are instead depend on the structured sparsity of the sparsity patterns,
accordingly encouraging a structured sparse $\widehat{\boldsymbol{\beta}}_{\mbox{Aggregate}}$.

We give extensions that cover a wide range possible structure inducing strategies. These include overlapping grouped structures and structural penalties based on hierarchical structures or arbitrary set functions. These parallel many convex methods for structured sparsity from the literature, see Section~\ref{structuedAgg}. Though structure can be useful for interpretability, we must consider whether injecting structure into a sparse method has a beneficial impact under reasonable conditions. In this paper we demonstrate that our estimators perform no worse than sparse estimators when the true model is structured sparse. In the group sparsity case, they can give dramatic improvements. These results hold for a very general class of structural modifications, including overlapping grouped structures.

We first give a review the sparse
aggregation method of~\cite{RigTsy10} in Section~\ref{ES}. In
Section~\ref{structuedAgg} we discuss our methods for structured
sparse aggregation.  We introduce two settings: structurally penalized
sparse aggregation (Section~\ref{structPenalizedNorm}), and group structured sparse
aggregation (Section~\ref{groupedEllZero}). We present the theoretical
properties of these estimators in Section~\ref{Theory}. We then
present a simulation study and an application to HIV drug resistance
in Section~\ref{applicationSimHIV}. We finally give some concluding remarks
and suggestions for future directions in Section~\ref{conclude}. Proofs of general versions of the main theoretical results are given in the supplementary material.

\section{Sparsity Pattern Aggregation}
\label{ES}

The sparse aggregation method of~\cite{RigTsy10} builds on the results
of~\cite{leungInfoMixing}. The
method creates an aggregate estimator from a weighted average the
$2^M$ ordinary least squares regressions on all subsets of the $M$
candidate covariates. The method encourages sparsity by letting the
weight in the average for a particular model increase as the sparsity
of the model increases.  We first establish our notation and
setting, and then present the basic formulas behind the method.  We
finally discuss its implementation  via a stochastic greedy algorithm. 

\subsection{Settings and the Sparsity Pattern Aggregation Estimator}

We consider the linear regression model:
\begin{align}
\textbf{y} = \textbf{X}^T\boldsymbol{\beta} + \boldsymbol{\epsilon}.
\end{align}
Here, we have a response $\textbf{y} \in \mathbb{R}^n$,
$n \times M$ data matrix $\textbf{X} =
[\textbf{x}_1,\ldots,\textbf{x}_M]$ --- where $\textbf{x}_i \in \mathbb{R}^n$, and vector of coefficients $\boldsymbol{\beta} \in
\textbf{R}^M$. From here on, we assume that $\textbf{X}$ is normalized so $\|\textbf{x}_i\|_2^2 \leq 1 \ \forall i$. The entries of the $n-$vector of errors
$\boldsymbol{\epsilon}$ are $\mbox{i.i.d.}$ $N(0,\sigma^2)$.  Assume that $\sigma^2$ is known.  Let
$\|\cdot\|_p$ dnote the $\ell_p$ norm for $p \geq 1$.  Let
$\mbox{supp}(\cdot)$ denote the support of a vector, the set of
indices for which the entries are nonzero.  Denote $\|\cdot\|_0 =
|\mbox{supp}(\cdot)|$ as the $\ell_0$ norm.  Let the set $\mathcal{I} = \{1,\ldots,M\}$ index the set of candidate covariates.

Define the set $\mathcal{P} = \{0,1\}^M$; $|\mathcal{P}| = |2^\mathcal{I}| = 2^M$. $\mathcal{P}$ encodes all sparsity
patterns over our set of candidate covariates
 ---
the $i$th element of $\textbf{p} \in \mathcal{P}$ is 1 if covariate
$i$ is included in the model, and 0 if it is excluded.  Let
$\widehat{\boldsymbol{\beta}}_\textbf{p}$ be the ordinary least squares solution
restricted to the sparsity pattern $\textbf{p}$:
\begin{align}
\label{OLS}
\widehat{\boldsymbol{\beta}}_\textbf{p} &= \underset{\boldsymbol{\beta} \in \mathbb{R}^M: \ 
  \mbox{supp}(\boldsymbol{\beta}) \subseteq \mbox{supp}(\textbf{p})}{\operatorname{argmin}}\| \textbf{y} - \textbf{X}\boldsymbol{\beta}\|^2_2.
\end{align}
Define the training error of an estimate $\widehat{\boldsymbol{\beta}}$ to be:
\begin{align}
\mbox{Error}(\widehat{\boldsymbol{\beta}}) = \|\textbf{y} - \textbf{X}\widehat{\boldsymbol{\beta}}\|^2_2.
\end{align}
Then, the sparsity pattern aggregate estimator coefficients are defined as: 
\begin{align}
\label{SPA-estimator}
\widehat{\boldsymbol{\beta}}^{SPA}:= \frac{\sum_{\textbf{p} \in \mathcal{P}}
  \widehat{\boldsymbol{\beta}}_\textbf{p} \exp\left( -\frac{1}{4\sigma^2} \mbox{Error}(\widehat{\boldsymbol{\beta}}_\textbf{p}) -
    \frac{\|\textbf{p}\|_0}{2} \right) \pi_\textbf{p}}{\sum_{\textbf{p}
    \in \mathcal{P}} \exp\left( -\frac{1}{4\sigma^2} \mbox{Error}(\widehat{\boldsymbol{\beta}}_\textbf{p}) - \frac{\|\textbf{p}\|_0}{2}
  \right) \pi_\textbf{p}}. 
\end{align}
Here, we obtain $\widehat{\boldsymbol{\beta}}^{SPA}$ by taking a weighted
average over all sparsity patterns $\textbf{p}$. The weights in this average
are a product of an exponentiated unbiased estimate of the risk and a
prior, $\pi_\textbf{p}$, over the sparsity patterns. This strategy is based on the work of~\cite{leungInfoMixing}, who demonstrated this form results in several appealing theoretical properties which form the basis of the theory of~\cite{RigTsy10} and our own methods. \cite{RigTsy10} consider the following prior:
\begin{align}
\label{SPA_PRIOR}
\pi_{\textbf{p}} := \left\{
     \begin{array}{lr}
       \frac{1}{H} \left( \frac{\|\textbf{p}\|_0}{2eM}
       \right)^{\|\textbf{p}\|_0} & \|\textbf{p}\|_0 \leq R\\
       \frac{1}{2} & \|\textbf{p}\|_0 = M\\
       0 & \mbox{else}
     \end{array}
   \right. .
 \end{align}
Here, $H$ is a normalizing constant and $R = \mbox{rank}(\textbf{X})$. The above prior places exponentially less weight on sparsity patterns
as their $\ell_0$ norm increases, up-weighting sparse models. The
weight of $1/2$ on the OLS solution is included for theoretical
calculations; in practice this case is treated as other cases --- see
the supplementary material. This
specific choice of prior had many theoretical and computational
advantages. In section~\ref{structuedAgg}, we consider modifications
to the prior weight to encourage both structure and sparsity. 

\subsection{Computation}
\label{stochasticAlgorithm}

Exact computation of the sparsity pattern aggregate estimator is
clearly impractical, since it would require fitting $2^M$ models. \cite{RigTsy10} give a
 Metropolis-Hastings stochastic greedy algorithm based on work by~\cite{PAC_bayes} for approximating the sparsity
pattern aggregate --- the procedure is reviewed in the supplement. The procedure performs a random walk over the hypercube of all
sparsity patterns.  Beginning with an empty model, in each step, one covariate is
randomly selected from the candidate set, and
proposed to be added to the model if it is already in the current model or to be removed from the current model
otherwise. These proposals are accepted or rejected using a Metropolis
step, with probability related to the product of the difference in risk and the ratio of prior weights. 

Two practical concerns arise from this approach. First, the algorithm assumes that $\sigma^2$ is known. Second, the metropolis algorithm requires significant additional computation than competing sparse methods. Regarding the variance,~\cite{RigTsy10} proposed a two stage scheme: the algorithm is run twice, and the residuals from the first run provide an estimate for the variance for the second run. To the second point, a simple analysis of the algorithm reveals that at each iteration of the MCMC method, we must fit a linear regression model. In order to effectively explore the sparsity pattern hypercube, we must run the Markov chain on the order of $M$, the number of candidate predictors, iterations. We can therefore expect computation times on the order of a linear regression fit times $M$. When $M$ is a much higher order than the number of observations, this is a concern. This makes the sparse estimator difficult to compute in very high dimensional settings. However, in a structured sparse problem, we may have structural information that effectively reduces the order of $M$, such as in group sparsity.

\section{Structured, Sparse Aggregation}
\label{structuedAgg}

The sparsity pattern aggregate estimator derives its sparsity property from
placing a prior on sparsity patterns that is
inversely proportional to the $\ell_0$ norm of the pattern.  This
up-weights models with sparsity patterns with low $\ell_0$ norm,
encouraging sparsity. We 
propose basing similar priors on different set functions than the
$\ell_0$ norm. These set functions are chosen so that the resulting estimator
simultaneously encourages sparsity and structure. Thus, the resulting
estimators upweight structured, sparse models.
We consider two class of
functions: structurally penalized $\ell_0$ norms and grouped $\ell_0$ norms.

\subsection{Penalized Structured Sparsity Aggregate Estimator}
\label{structPenalizedNorm}

Consider penalizing the $\ell_0$ norm with some non-negative set function that measures the structure of the sparsity pattern.  We will show later (see
Assumption~\ref{assumptionLOS} in Section~\ref{TheoryS1}) that if this set function if non-negative and does not exceed $M$, we can guarantee similar theoretical properties as
the sparsity pattern aggregate estimator.  More formally, consider the
following extension:  
\begin{align}
\textbf{p} \in \mathcal{P}: \|\textbf{p}\|_{0,c} &:= \|\textbf{p}\|_0 + \|\textbf{p}\|_c,\\
\mbox{where: }\|\textbf{p}\|_c := \|\mbox{supp}(\textbf{p})\|_c: 2^\mathcal{I} &\to [0,M] \subset \mathbb{R},\\
\|\textbf{0}\|_c & := 0.
\end{align}
We then define the following prior on $\mathcal{P}$:
\begin{align}
\label{SSPA_PRIOR}
\pi_{\textbf{p},c} := \left\{
     \begin{array}{lr}
       \frac{1}{H_c} \left( \frac{\|\textbf{p}\|_{0,c}}{2eM}
       \right)^{\|\textbf{p}\|_{0,c}} & \|\textbf{p}\|_{0} \leq R\\
       \frac{1}{2} & \|\textbf{p}\|_{0} = M\\
       0 & \mbox{else}
     \end{array}
   \right. .
 \end{align}
Where $H_c$ is a normalizing constant. For our subsequent theoretical analysis, we note that since
$\|\textbf{p}\|_{0,c} \leq 2M$ then we know that $H_c
\leq 4$. We then define the structured sparsity aggregate (SSA) estimator as: 
\begin{align}
\widehat{\boldsymbol{\beta}}^{SSA}:= \frac{\sum_{\textbf{p} \in \mathcal{P}}
  \widehat{\boldsymbol{\beta}}_\textbf{p} \exp\left( -\frac{1}{4\sigma^2} \mbox{Error}(\widehat{\boldsymbol{\beta}}_\textbf{p}) -
    \frac{\|\textbf{p}\|_0}{2} \right) \pi_{\textbf{p},c}}{\sum_{\textbf{p}
    \in \mathcal{P}} \exp\left( -\frac{1}{4\sigma^2} \mbox{Error}(\widehat{\boldsymbol{\beta}}_\textbf{p}) - \frac{\|\textbf{p}\|_0}{2}
  \right) \pi_{\textbf{p},c}}. 
\end{align}
We now discuss some possible choices for the structural penalty
$\|\cdot\|_c$.  Note that the general consequence of the prior is that
sparsity patterns with higher values of $\|\cdot\|_c$ will be
down-weighted.  At the same time, the prior still contains the
$\ell_0$ norm as an essential element, and so it enforces a trade-off between
sparsity and the structure captured by the additional term.   

\begin{itemize}
\item \textbf{Covariate Weighting}.
Consider the function $\|\textbf{p}\|_c = \sum_{i=1}^M
c_i\textbf{p}_i$ such that $\sum_{i=1}^M c_i < M$, $c_i>0 \ \forall \
i$.  This has the effect of weighting the covariates, discouraging those with
high weight to enter the model. These weights can be determined in a
wide variety of ways, including simple prior belief elicitation.  This weighting scheme is related to the
 prior suggested in~\cite{Hoeting&1999} in the bayesian model averaging
setting. This strategy also has the flavor of the individual weighting in the adaptive lasso, where \cite{Zou_2006} considered weighting each coordinate in the lasso using coefficient estimates from OLS or marginal regression.

\item \textbf{Graph Structures}.
Generalizing previous work, \cite{Bach10structuredsparse} suggested many
structure inducing set functions in the regularization setting.  Many of
these functions can be easily adapted to this framework. For example, given a directed acyclic graph (DAG)
structure over $\mathcal{I}$, the following penalty encourages a
hierarchical structure:  
\begin{align}
\|\textbf{p}\|_c = |\{\mbox{Ancestors of supp}(\textbf{p})\}|.
\end{align}
If we desire strong hierarchy, we can additionally define $\pi_{\textbf{p},c} := 0$ if the sparsity pattern of $\textbf{p}$ does not obey the hierarchical structure implied by the DAG. Strong hierarchy may also greatly increase the speed of the MCMC algorithm by restricting the number of predictors potentially sampled at each step.

Alternately, suppose we have a set of weights over pairs of predictors
represented by the function $d: \mathcal{I} \times \mathcal{I} \to
\mathbb{R}^+$. Given a graph over the candidate covariates, this could correspond to edge weights, or
the shortest path between two nodes (covariates). More generally, it could
correspond to a natural geometric structure such as a line or a
lattice, see~\cite*{percivalStructureSparse} for such
an example. We can use these weights to define the cut function: 
\begin{align}
\|\textbf{p}\|_c = \sum_{i \in \mbox{supp}(\textbf{p}); \ j \notin \mbox{supp}(\textbf{p})} d(i,j).
\end{align}
This encourages sparsity patterns to partition the set $\mathcal{I}$
into two maximally disconnected sets, as defined by low values of
$d(\cdot,\cdot)$. This would give sparsity patterns corresponding
to isolated neighborhoods in the graph.

\item \textbf{Cluster Counting}.
We finally propose a new $\|\cdot\|_c$ that measures the
structure of the sparsity pattern by counting the number of clusters
in $\textbf{p}$.  Suppose, we now have a symmetric weight function $d:
\mathcal{I} \times \mathcal{I} \to \mathbb{R}^+$. Suppose we also set
a constant $h>0$. Then, we count the clusters using the following
procedure: 
\begin{enumerate}
\item Define the fully connected weighted graph over the set supp$(\textbf{p})$ with weights given by $d(\cdot,\cdot)$.
\item Break all edges with weight great than $h$.
\item Return the remaining number of connected components as $\|\textbf{p}\|_c$.
\end{enumerate}
This definition encourages sparsity patterns that are clustered with
respect to $d(\cdot,\cdot)$. For computational considerations, note
that this strategy is the same as single linkage clustering with
parameter $h$, or building a minimal spanning tree and breaking all
edges with weight greater than $h$. For many geometries, this definition of
$\|\cdot\|_c$ is easy to compute and update for the MCMC algorithm.  
\end{itemize}

\subsection{Group Structured Sparsity Aggregate Estimator}
\label{groupedEllZero}

In the framework of structured sparsity, one popular representation
of structure is via groups of variables, cf.~\cite{Yuan06modelselection}.  For example, a factor 
covariate with $u$ levels, as in an ANOVA model, can be represented as a collection of
$u-1$ indicator variables.  We would not
select these variables individually, instead preferring to include or
exclude them as a group.  In the case where these groups partition
$\mathcal{I}$, this structure can be easily incorporated into the prior,
theory, and implementation of the sparsity pattern aggregate
estimator.   Suppose we a priori define: 
\begin{align}
\mathcal{G} &:= \{g\}\ \mbox{such that } g\subset \mathcal{I} \ \forall g; \mbox{and }  \cup_{g \in \mathcal{G}} g = \mathcal{I}, \forall g, g' \in \mathcal{G}, g \cap g' = \emptyset, \\
\|\boldsymbol{\beta}\|_{0,\mathcal{G}} &:= |\{g: g \cap \mbox{supp}(\boldsymbol{\beta}) \neq \emptyset \}|,\\
\|\boldsymbol{\beta}\|_{1,\mathcal{G}} &:= \sum_{g \in \mathcal{G}} \|\boldsymbol{\beta}_g\|_2.
\end{align}
$\|\boldsymbol{\beta}\|_{1,\mathcal{G}}$ is the same as the grouped lasso
penalty~\citep{Yuan06modelselection}, which is used to induce sparsity
at the group level in the 
regularization setting. $\|\boldsymbol{\beta}\|_{0,\mathcal{G}}$ is simply the number of groups needed to cover the sparsity pattern of $\boldsymbol{\beta}$.
Thus, we have simply replaced sparsity patterns over all subsets of
predictors with sparsity patterns over all subsets of groups of
predictors.  We can show that the theoretical framework of~\cite{RigTsy10}  
holds with $\mathcal{P} =
\{0,1\}^{|\mathcal{G}|}$, $\|\boldsymbol{\beta}\|_0$ replaced with $\|\boldsymbol{\beta}\|_{0,\mathcal{G}}$,
and $\|\boldsymbol{\beta}\|_1$ replaced with $\|\boldsymbol{\beta}\|_{1,\mathcal{G}}$.
 
A more interesting and flexible case arises when we allow the elements
of $\mathcal{G}$ to overlap.  Here, we adopt the framework
of~\cite{Jacob:2009:GLO:1553374.1553431}, who gave a norm and penalty
for inducing sparsity patterns using overlapping groups in the regularization
setting. In this case, we define the groups as any collection of sets of covariates:
\begin{align}
\mathcal{G} &:= \{g\} \ \mbox{such that } g\subset \mathcal{I} \ \forall g; \mbox{and } \cup_{g \in \mathcal{G}} g = \mathcal{I}.
\end{align}
We now define the $\mathcal{G}$-decomposition as the following set
of size $|\mathcal{G}|$: 
\begin{align}
\mathcal{V}_\mathcal{G}(\boldsymbol{\beta}) &= \{\textbf{v}_g: g
\in \mathcal{G}, \textbf{v}_g \in \mathbb{R}^M \mbox{ s.t. }  \mbox{supp}(\textbf{v}_g) \subseteq g
\},\\
 \mbox{such that} & \sum_{\textbf{v}_g \in \mathcal{V}_\mathcal{G}(\boldsymbol{\beta})} \textbf{v}_g = \boldsymbol{\beta}.
\end{align}
That is, $\mathcal{V}_\mathcal{G}(\boldsymbol{\beta})$ contains single $\textbf{v}_g$ for each $g \in \mathcal{G}$. For arbitrary $\mathcal{G}$ and
$\boldsymbol{\beta}$, $\mathcal{V}_\mathcal{G}(\boldsymbol{\beta})$ is not
unique. We then define the following functions, analogous to the $\ell_0$
and $\ell_1$ norms of the usual sparsity framework: 
\begin{align}
\|\boldsymbol{\beta}\|_{0,\mathcal{G}} &= \min_{G \subset \mathcal{G}; \cup_{g \in
    G} g = \mbox{supp}(\boldsymbol{\beta})} |G|, \label{groupedEllZeroNorm}\\
\|\boldsymbol{\beta}\|_{1,\mathcal{G}} &= \min_{\mathcal{V}_\mathcal{G}(\boldsymbol{\beta})}\left(
  \sum_{g \in \mathcal{G}} \|\textbf{v}_g\|_2 \right) \label{groupedEllOneNorm}
\end{align}
In $\|\cdot\|_{1,\mathcal{G}}$, the minimum is over all possible decomposition $\mathcal{V}_\mathcal{G}(\cdot)$.  Computing $\|\cdot\|_{0,\mathcal{G}}$ is difficult for
arbitrary $\mathcal{G}$.  However, in most applications $\mathcal{G}$
has some regular structure which allows for efficient computation.   The norm
in Equation~\ref{groupedEllZeroNorm} leads to the following choice of prior 
on $\mathcal{P}$: 
\begin{align}
\label{SSPA_GROUPED_PRIOR}
\pi_{\textbf{p}, \mathcal{G}} := \left\{
     \begin{array}{lr}
       \frac{1}{H_ \mathcal{G}} \left( \frac{\|\textbf{p}\|_{0,\mathcal{G}}}{2e|\mathcal{G}|}
       \right)^{\|\textbf{p}\|_{0, \mathcal{G}}} & \|\textbf{p}\|_{0} \leq R\\
       \frac{1}{2} & \|\textbf{p}\|_{0} = M\\
       0 & \mbox{else}
     \end{array}
   \right. .
 \end{align}
By considering
all unions of groups, we obtain an upper bound for the normalizing constant $H_\mathcal{G}
\leq 4$.  We then define the grouped sparsity aggregate (GSA) estimator as:
\begin{align}
\widehat{\boldsymbol{\beta}}^{GSA}:= \frac{\sum_{\textbf{p} \in \mathcal{P}}
  \widehat{\boldsymbol{\beta}}_\textbf{p} \exp\left( -\frac{1}{4\sigma^2} \mbox{Error}(\widehat{\boldsymbol{\beta}}_\textbf{p}) -
    \frac{\|\textbf{p}\|_0}{2} \right) \pi_{\textbf{p}, \mathcal{G}}}{\sum_{\textbf{p}
    \in \mathcal{P}} \exp\left( -\frac{1}{4\sigma^2} \sum_{i=1}^n \mbox{Error}(\widehat{\boldsymbol{\beta}}_\textbf{p}) - \frac{\|\textbf{p}\|_0}{2}
  \right) \pi_{\textbf{p}, \mathcal{G}}}. 
\end{align}

We leave $\mathcal{G}$ general throughout this section and the
subsequent theoretical analysis. There are many possible
definitions of $\mathcal{G}$, such as connected components
or neighborhoods in a graph, groups of factor predictors, or application
driven groups --- see~\cite{Jacob:2009:GLO:1553374.1553431} for some
examples.  In particular, many of the structures mentioned in
Section~\ref{structPenalizedNorm} can be encoded as a series of
groups. 


\section{Theoretical Properties}
\label{Theory}

\cite{RigTsy10} showed that the sparsity pattern aggregate estimator enjoyed great theoretical properties. In summary, they showed that the estimator adapted to the sparsity of the target, measured in both the $\ell_0$ and $\ell_1$ norm. Further, they showed that their sparsity oracle inequalities were optimal in a minimax sense, in particular superior to rates obtained for popular estimators such as the lasso. Moreover, their results required fewer assumptions than those of the lasso, cf~\cite{Bickel_simultaneousanalysis}. In the supplementary material, we give a theoretical
framework for aggregation using priors of our form ---
Equation~\ref{SSPA_PRIOR} and~\ref{SSPA_GROUPED_PRIOR}. The following shows specific applications of this theory, yielding a set of structured sparse oracle inequalities, the first of their kind. 


\subsection{Structurally Penalized $\ell_0$ Norm}
\label{TheoryS1}

We first state an assumption:
\begin{asm}
\label{assumptionLOS}
For all $\textbf{p} \in \mathcal{P}$ where $R>\|\textbf{p}\|_0>0$:
\begin{align}
\frac{\|\textbf{p}\|_0}{\|\textbf{p}\|_{0,c}} \leq
\log\left( 1 + \frac{eM}{\max(\|\textbf{p}\|_{0,c},1)}\right).
\end{align}
\end{asm}
Numerical analysis reveals that a sufficient condition for this assumption is $0 \leq \|\textbf{p}\|_c \leq M$.
%
\begin{prp} 
\label{prop1pen}
Suppose Assumption~\ref{assumptionLOS} holds. For any $M\geq1, n \geq1$, the structured sparsity aggregate
estimator satisfies:
\begin{align}
\mathbb{E} \|\textbf{X}{\widehat{\boldsymbol{\beta}}^{SSA}} - \textbf{y} \|^2_2 \leq 
\min_{\boldsymbol{\beta} \in \mathbb{R}^M} \left\{ \|\textbf{X}{\boldsymbol{\beta}} - \textbf{y}\|^2_2 + \min
  \left\{\frac{\sigma^2 R}{n}, \mbox{ } 9 \sigma^2
  \frac{M_c(\boldsymbol{\beta})}{n} \log \left( 1 + 
    \frac{eM}{\max(M_c(\boldsymbol{\beta}), 1)}\right) \right\} \right\} +
\frac{8\sigma^2}{n}\log 2
\end{align}
Here, $R = \mbox{rank}(\textbf{X})$, and $M_c(\boldsymbol{\beta}) = \|\mbox{sparsity}(\boldsymbol{\beta})\|_{0,c}$, where $\mbox{sparsity}(\boldsymbol{\beta})$ is the sparsity pattern of $\boldsymbol{\beta}$.
\end{prp}

A key property of the next proposition is the existence of some $\gamma \geq 1$ such that $\forall \textbf{p} \in \mathcal{P}: \|\textbf{p}\|_0 \leq \|\textbf{p}\|_{0,c} \leq \gamma \| \textbf{p}\|_{0}$.

\begin{prp}
\label{prop2pen}
Suppose Assumption~\ref{assumptionLOS} holds. 
Suppose the structural penalty in structured sparsity aggregate (SSA) estimator satisfies $\forall \textbf{p} \in \mathcal{P}: \|\textbf{p}\|_0 \leq \|\textbf{p}\|_{0,c} \leq \gamma \| \textbf{p}\|_{0}$ for some $\gamma \geq 1$.  Then for any
$M\geq 1, n \geq 1$, the SSA estimator
satisfies:
\begin{align}
 \mathbb{E}\|\textbf{X}{\widehat{\boldsymbol{\beta}}^{SSA}} - \textbf{y} \|^2_2 \leq \min_{\boldsymbol{\beta} \in
   \mathbb{R}^M} \{ \|\textbf{X}{\boldsymbol{\beta}} - \textbf{y}\|^2_2 + \phi_{n,M}(\boldsymbol{\beta}) \} +
 \frac{\sigma^2}{n}(9\log(1+eM) + 8\log2)
\end{align}
where $\phi_{n,M}(0) := 0$ and for $\boldsymbol{\beta} \neq 0$:
\begin{align}
\phi_{n,M} = \min\left[ \frac{\sigma^2}{n},\frac{9\sigma^2
    M_c(\boldsymbol{\beta})}{n} \log \left( 1+
    \frac{eM}{\max(M_c(\boldsymbol{\beta}),1)}\right), \frac{11\sigma
    \sqrt{\gamma}\|\boldsymbol{\beta}\|_1}{\sqrt{n}}\sqrt{\log \left( 1 +
      \frac{3eM\sigma}{\|\boldsymbol{\beta}\|_1\sqrt{\gamma n}}\right)}  \ \right]
\end{align}
\end{prp}

\subsection{Grouped $\ell_0$ Norm}
\label{TheoryS2}

We first state an
Assumption:
\begin{asm}
\label{assumptionGRPLO}
For all $\textbf{p} \in \mathcal{P}$ where $R>\|\textbf{p}\|_0>0$:
\begin{align}
\frac{\|\textbf{p}\|_0}{\|\textbf{p}\|_{0,\mathcal{G}}} \leq
\log\left( 1 + \frac{e|\mathcal{G}|}{\max(\|\textbf{p}\|_{0,\mathcal{G}},1)}\right).
\end{align}
\end{asm}
This assumption does not hold uniformly for all sparsity patterns and for all choices of
$\mathcal{G}$. A sufficient condition for the assumption is:
\begin{align}
\max_{g \in \mathcal{G}} |g| \leq \log(1 + e|\mathcal{G}|/R).
\end{align}
In particular, for sparsity patterns with low $\ell_0$ norm relative to $M$, the
assumption is satisfied provided the cardinality of $\mathcal{G}$ is large enough.  



\begin{prp} 
\label{prop1group}
Suppose Assumption~\ref{assumptionGRPLO} holds. For any $M\geq1, n \geq1$, the grouped sparsity aggregate
estimator satisfies:
\begin{align}
\mathbb{E} \|\textbf{X}{\widehat{\boldsymbol{\beta}}^{GSA}} - \textbf{y} \|^2_2 \leq 
\min_{\boldsymbol{\beta} \in \mathbb{R}^M} \left\{ \|\textbf{X}{\boldsymbol{\beta}} - \textbf{y}\|^2_2 + \min
  \left\{\frac{\sigma^2 R}{n}, \mbox{ } 9 \sigma^2
  \frac{M_\mathcal{G}(\boldsymbol{\beta})}{n} \log \left( 1 + 
    \frac{e|\mathcal{G}|}{\max(M_\mathcal{G}(\boldsymbol{\beta}), 1)}\right) \right\} \right\} +
\frac{8\sigma^2}{n}\log 2
\end{align}
Here, $R = \mbox{rank}(\textbf{X})$, and $M_\mathcal{G}(\boldsymbol{\beta}) = \|\mbox{sparsity}(\boldsymbol{\beta})\|_{0,\mathcal{G}}$, where $\mbox{sparsity}(\boldsymbol{\beta})$ is the sparsity pattern of $\boldsymbol{\beta}$.
\end{prp}

\begin{prp}
\label{prop2group}
Suppose Assumption~\ref{assumptionGRPLO} holds. 
Then for any
$M\geq 1, n \geq 1$, the grouped sparsity aggregate estimator
satisfies:
\begin{align}
 \mathbb{E}\|\textbf{X}{\widehat{\boldsymbol{\beta}}^{GSA}} - \textbf{y} \|^2_2 \leq \min_{\boldsymbol{\beta} \in
   \mathbb{R}^M} \{ \|\textbf{X}{\boldsymbol{\beta}} - \textbf{y}\|^2_2 + \phi_{n,\mathcal{G}}(\boldsymbol{\beta}) \} +
 \frac{\sigma^2}{n}(9\log(1+e|\mathcal{G}|) + 8\log2)
\end{align}
where $\phi_{n,\mathcal{G}}(0) := 0$ and for $\boldsymbol{\beta} \neq 0$:
\begin{align}
\phi_{n,\mathcal{G}} = \min\left[ \frac{\sigma^2}{n},\frac{9\sigma^2
    M_\mathcal{G}(\boldsymbol{\beta})}{n} \log \left( 1+
    \frac{e|\mathcal{G}|}{\max(M_\mathcal{G}(\boldsymbol{\beta}),1)}\right), \frac{11\sigma
    \|\boldsymbol{\beta}\|_{1,\mathcal{G}}}{\sqrt{n}}\sqrt{\log \left( 1 +
      \frac{3e|\mathcal{G}|\sigma}{\|\boldsymbol{\beta}\|_{1,\mathcal{G}}\sqrt{n}}\right)}  \ \right]
\end{align}
\end{prp}

\subsection{Discussion of the Results}

For each class of prior, we give two main results.  The first result
shows that each procedure enjoys adaptation in terms of 
the appropriate structured sparsity measuring set function --- $\|\cdot\|_{0,c}$ and
$\|\cdot\|_{0,\mathcal{G}}$, respectively.  The bound is thus best when the structured sparsity of the regression function is small, as measured by the appropriate set functions; the estimator adapts to the structured sparsity of the target. The second demonstrates that the estimators also adapts to structured sparsity measured in terms of a corresponding convex norm ---
$\|\cdot\|_1$ and $\|\cdot\|_{1,\mathcal{G}}$.  This is useful when some entries
of $\boldsymbol{\beta}$ contribute little to the convex norm, but still
incur a penalty in the corresponding set function. For
example, a small isolated entry of $\boldsymbol{\beta}$ contributes little to the $\ell_1$
norm, but is heavily weighted in the structurally
penalized $\ell_0$ norm.  

Comparing the results to the corresponding results in~\cite{RigTsy10}, these results reveal some benefits and drawbacks to adding structure to the sparse aggregation procedure. In the penalized case, the results show that the structured estimator enjoys the same rates as the sparse estimator when the penalty is low. When structure is not present in the target, the sparse estimator is superior, as expected. Proposition~\ref{prop2pen} is still given in terms of the $\ell_1$ norm, which only measures sparsity. The price for adding structure to the procedure appears in the additional factor of $\sqrt{\gamma}$. While these results are not dramatic, the previous discussion (Section~\ref{structPenalizedNorm}) and subsequent simulation study (Section~\ref{applicationSimHIV-sim}) show that the penalized version is flexible and powerful in practice. 

In the grouped case, the results are more appealing. Since the grouped $\ell_0$ and $\ell_1$ norms are potentially much smaller than their ungrouped counterparts, the results here give better constants than their sparse versions. These improvements may be dramatic: previous work on the grouped lasso, cf~\cite{Lounici_takingadvantage},~\cite{zhang_benefits}, revealed great benefits to grouped structures. Following the settings of~\cite{Lounici_takingadvantage}, consider a multi-task regression setting in which we desire the same sparsity pattern across tasks.  Then, if the number of tasks is on the same order or of a higher greater than the number of samples per task ($n$), a grouped aggregation approach would reduce the order (in $n$) of the rates in the theoretical results. We can also expect such improvements for an overlapping set of groups that do not highly overlap.

%

The propositions given in the previous subsections are simplified
versions of those proved for the sparsity pattern aggregate estimator
in~\cite{RigTsy10}. We note that the full results can be extended to our
estimators, we omit the derivation for brevity. In addition
to these more complex statements, \cite{RigTsy10} also gave a detailed
theoretical discussion of these results in comparison to the lasso and
 BIC aggregation estimators~\cite{buneaAggregation}, concluding that their estimator enjoyed superior and near optimal rates. Since our rates differ by no more than constants when the target is truly structured and sparse, we conclude that in such settings a structured approach can give great benefits.

\section{Applications}
\label{applicationSimHIV}

\subsection{Simulation Study}
\label{applicationSimHIV-sim}

We now turn to a simulation study.  ~\cite{RigTsy10}
presented a detailed simulation study comparing their sparsity pattern
aggregate estimator --- see Section~\ref{ES} --- to
numerous sparse regression methods. They demonstrated that the sparsity
pattern aggregate was superior to the competitor methods. Therefore,
we primarily compare our technique to the sparsity pattern aggregate estimator. We will show that the structured sparsity pattern aggregate estimator is superior under appropriate settings where the target is structured.

For brevity, we consider only the structurally penalized $\ell_0$ norm.  In the following, we employ
our cluster counting penalty, described in
Section~\ref{structPenalizedNorm}, with $h=3$.  We consider two settings that offer natural geometries and notions of structure: connected
components in a line structure (see, e.g. the top
left display in Figure~\ref{oneDimOneRunFigure}), and blocks in a
two-dimensional lattice (see, e.g. the top left
display in Figure~\ref{twoDimOneRunFigure}). Using these natural geometries, we let $d(\cdot, \cdot)$ be euclidean distance. We uniformly at random set the appropriate entries of a
true coefficient vector $\overline{\boldsymbol{\beta}}$ to be one of $\{+1,-1\}$. Each entry
of the $n \times M$ design matrix $\textbf{X}$ are independent
standard random normal variables.  We additionally generate a $n
\times M$ matrix $\textbf{X}_{test}$ to measure prediction
performance, see below. We consider different values of $n$ --- the
number of data points, $M$ --- the number of candidate covariates; represented as columns in $\textbf{X}$, $C$
--- the number of
clusters as measured by the cluster counting penalty applied to the true sparsity pattern, and $C_{on}$ --- the
number of nonzero entries per cluster in $\overline{\boldsymbol{\beta}}$.  We enforce non overlapping clusters giving $\|\overline{\boldsymbol{\beta}}\|_0 = C\times C_{on}$.  For direct comparison we follow~\cite{RigTsy10}, and set the noise level $\sigma = \|\overline{\boldsymbol{\beta}}\|_0/9$,
and run the MCMC algorithm for 7000 iterations, discarding the first 3000. We repeat each simulation setting 250 times.

We use two metrics to
measure performance. First, prediction risk:
\begin{align}
\mbox{Prediction}(\widehat{\boldsymbol{\beta}}) := \frac{\|\textbf{X}_{\mbox{test}}^T(\overline{\boldsymbol{\beta}} - \widehat{\boldsymbol{\beta}})\|_2^2}{n}.
\end{align}
Our second metric measures the estimation of $\overline{\boldsymbol{\beta}}$:
\begin{align}
\mbox{Recovery}(\widehat{\boldsymbol{\beta}}) := \frac{\|\overline{\boldsymbol{\beta}} - \widehat{\boldsymbol{\beta}}\|_2^2}{\|\overline{\boldsymbol{\beta}}\|_2^2}.
\end{align}
In each of the above, $\widehat{\boldsymbol{\beta}}$ denotes some estimate of
$\overline{\boldsymbol{\beta}}$.  We compare
against our structured sparsity aggregate estimator (SSA) against the
sparsity pattern aggregate estimator (SPA) and the lasso
(lasso) --- note that the true coefficients, while clustered, are not smooth, making these settings inappropriate applications for structured smooth estimators such as the 1d or 2d fused lasso~\citep*{fusedLasso}.  For the lasso, we we choose the tuning parameter
$\lambda$ using 10-fold cross validation, and refit the model using ordinary least squares
regression, both within and outside cross validation. This strategy effectively uses the lasso only for its variable
selection properties and avoids shrinkage in $\widehat{\boldsymbol{\beta}}$.  We employ the
R package {\tt glmnet}~\citep{glmnetRpackage} to fit the lasso.

Tables~\ref{penalizedEllZeroSimulationOneD}
and~\ref{penalizedEllZeroSimulationTwoD} display the results. In all cases, the structured sparse estimator is
superior to the sparse estimator, and both methods are superior
to the lasso.  Although the mean prediction and recovery for the
aggregation estimators are within two standard errors of each
other, for paired runs on the same simulated data set, the structured
sparse estimator is superior in both metrics at least 95\% of the
time, for all settings.  Figures~\ref{oneDimOneRunFigure} and~\ref{twoDimOneRunFigure} display results for a sample sparsity pattern in both settings.  We can clearly see the superiority of the aggregation methods over the lasso. In both figures, we see that both aggregation methods correctly estimated the true sparsity pattern. However, in the sparse estimator, the Markov chain spent many iterations adding and dropping covariates far away from the true clusters. This did not happen in the structured estimators, giving a much sharper picture of the sparsity pattern in both cases. Rejecting these wandering steps gave the structured estimator better numerical performance in both prediction and estimation.

\subsection{Application to HIV Drug Resistance}
 
We now explore a data application which calls for a structured
sparse approach. Standard drug therapy for Human Immunodeficiency Virus (HIV) inhibits the activity of proteins produced by the
virus. HIV is able to change its protein structure easily and
become resistant to the drugs. The goal is then to determine which mutations drive this
resistance. We use regression to determine
the relationship between a particular strain of HIV's resistance to a drug
 and its protein sequence. \cite{Soo} studied this
problem using sparse regression techniques.

Casting this problem as linear regression, the continuous response
is drug resistance, measured by $\log$ dosage of the drug needed to effectively
negate the virus' reproduction.  The covariates derive from the protein
sequences. Each sequence is 99 amino acids long, so we view
each of these 99 positions as factors. Breaking each of
these factors into levels, we obtain mutation covariates, which is our set of candidate predictors. If a location displays $A$ different amino acids across the data, we obtain $A-1$ mutation covariates.  Thus, each
covariate is an indicator variable for the occurrence of a particular
amino acid at a particular location in the protein sequence. Note that
many positions in the protein sequence display no variation throughout the data set --- these positions always display the same amino acid --- and are
therefore dropped from the analysis.
In summary, the predictors are mutations in the sequence, and
the response is the log dosage. A sparse model would show exactly which mutations are most important in driving resistance.  We are interested in which mutations predict drug resistance, rather than only which locations predict dug resistance.  Therefore, we do not select the mutation covariates from a location jointly.  We instead treat each mutation separately.

Additional biological information gives us reason to
believe a structured, sparse model is more appropriate. Proteins
typically function by active sites.  That is, localized areas of the
protein are more important to the protein function than
others. Viewing the sequence as a simple linear structure, we expect
that selected mutations should occur clustered in this
structure. We can cluster the mutations by defining a distance in straightforward way: since
each mutation covariate is also associated with a location, we can
define $d(\cdot, \cdot)$, the distance between a pair of mutation covariates, as the absolute
difference in their locations.  

We apply our structured sparse aggregation (SSA) method along with sparse
aggregation (SPA), forward stepwise regression, and the lasso to the data for
drug Saquinavir (SQV) ---
see~\cite*{database} for details on the data
and~\cite{percivalStructureSparse} for another structured sparse
approach to the analysis; the data are available as a data set in the R
package {\tt BLINDED}~\cite{casparRpackage}. We set $h=3$ in
our cluster counting structural penalty for the structured aggregation
method.

We display a comparison of the sparsity patterns for the
methods in Figure~\ref{hivFigure}.  We see that each method selects similar mutations.  As expected, the structured sparse estimator encourages clustered selection of mutations, giving us two clear important regions.  In contrast, the sparse aggregation estimator, stepwise regression, and the lasso suggest mutations across the protein sequence.

We finally evaluate the predictive performance of the four methods using data splitting. We split the
data into three equal groups, and compare the mean test error from
using each set of two groups as a training set, and the third as a
test set. Table~\ref{hivTable} shows that both aggregation estimators
are superior to the lasso and stepwise regression.  Although the mean
test error is lower for the sparse aggregation estimator, it is within
a single standard deviation of the structured estimator's mean test
error. Therefore, the structured estimator gives comparable predictive
power, with the extra benefit of superior biological interpretability. 


\section{Conclusion}
\label{conclude}

In this paper, we proposed simple modifications of a powerful sparse aggregation
technique, giving a framework for structured
sparse aggregation. We presented methods for two main classes of structured sparsity: set
function based structure and grouped based structure.  These
aggregation estimators place highest on weight models whose sparsity
patterns are the most sparse and structured. 
We showed that these estimators enjoy appropriate oracle
inequalities --- they adapt to the structured sparsity of the targets. Further, we showed that in practice these methods are effective in the appropriate setting.

In the theory throughout this paper, we
considered a particular structure in the prior in order to easily
compare theoretical properties with sparse estimators. In practice,
the form of the prior may be modified further. For example, we need
not restrict our structural penalty to be less than the number of
predictors. In our current formulation, this restriction forced us to
consider sparsity and structure with equal weight. 

Although both the sparsity pattern and structured sparsity pattern estimators display good promise theoretically and in practice, there are several practical challenges remaining.  First, while~\cite{RigTsy10} suggested a strategy for dealing with the assumption that $\sigma^2$ is known, it requires running another Markov chain to find a good estimate for $\sigma^2$. This strategy is slow, and the stochastic greedy algorithm is much slower than comparable sparse techniques.  While the algorithm is not prohibitively slow, speedups would greatly enhance its utility. Currently, the algorithm must be run for at least approximately $10 \times M$ iterations so that it is time to search over all $M$ covariates. Since each iteration requires an OLS regression fit, if $M$ is of the same or greater order than $n$, this is a significant drawback. Thus, the estimator does not scale well to high dimensions. In future work, we can also consider a specialized version of
the stochastic greedy algorithms adapted to our structured priors. 

\appendix
\section{Implementation of Aggregation Estimators}

\subsection{Metropolis Algorithm}

Here, we give the implementation of the sparsity pattern aggregation estimator, proposed by~\cite{RigTsy10}. This approach can be naturally adapted to the structured case. For numerical implementation, \cite{RigTsy10} consider the following simplified prior:
\begin{align}
\label{SPA_PRIOR_used}
\pi_{\textbf{p}}^* := \left\{
     \begin{array}{lr}
       \frac{1}{H^*} \left( \frac{||\textbf{p}||_0}{2eM}
       \right)^{||\textbf{p}||_0} & ||\textbf{p}||_0 \leq R\\
       0 & \mbox{else}
     \end{array}
   \right. .
 \end{align}

Initialize the algorithm by setting $\textbf{p}(1) = \textbf{0} \in
\mathcal{P}$.  Repeat the following steps for $t = 1,\ldots,T$.

\begin{enumerate}
\item Generate a random integer $i$ in the set $\{1,2,\ldots,M\}$ from a
  discrete uniform distribution. Set the proposal sparsity
  pattern  $\textbf{q}(t)$ as $\textbf{p}(t)$ with entries
  satisfying:
  \begin{align}
  \textbf{q}(t) := \left\{
    \begin{array}{lr}
      \textbf{p}(t)_j  & i \neq j\\
      1 - \textbf{p}(t)_j & i = j
    \end{array}
  \right.
\end{align}
That is, entry $i$ has been toggled from ``on'' to ``off'', or visa versa.
\item Compute $\widehat{\boldsymbol{\beta}}_{\textbf{p}(t)}$ and
  $\widehat{\boldsymbol{\beta}}_{\textbf{q}(t)}$, the least squares estimators under
  sparsity patterns $\textbf{p}(t)$ and $\textbf{q}(t)$, respectively. Let:
\begin{align}
r(t) &= \min\left( \frac{\nu_{\textbf{q}(t)}}{\nu_{\textbf{p}(t)}},1
\right),\\
\frac{\nu_{\textbf{q}(t)}}{\nu_{\textbf{p}(t)}}& = \exp \left(
  \frac{1}{4\sigma^2} \left(\mbox{Error}\left(\widehat{\boldsymbol{\beta}}_{\textbf{p}(t)}\right) -
  \mbox{Error}\left(\widehat{\boldsymbol{\beta}}_{\textbf{q}(t)}\right)\right)  + \frac{||\textbf{p}(t)||_0 -
    ||\textbf{q}(t)||_0}{2} \right) 
\frac{\pi_{\textbf{q}(t)}^*}{\pi_{\textbf{p}(t)}^*}.
\end{align}
Here, for the prior in Equation~\ref{SPA_PRIOR_used}:
\begin{align}
\frac{\pi_{\textbf{q}(t)}^*}{\pi_{\textbf{p}(t)}^*} &= 
\left( 1 +  \frac{||\textbf{q}(t)||_0 -
      ||\textbf{p}(t)||_0}{||\textbf{p}(t)||_0}\right)^{||\textbf{q}(t)||_0}
\left( \frac{||\textbf{p}(t)||_0}{2eM}\right) ^{||\textbf{q}(t)||_0 -
  ||\textbf{p}(t)||_0}. \label{priorRatio}
\end{align}
\item Update $\textbf{p}(t)$ by generating the following random variable:
\begin{align}
\textbf{p}(t+1) := \left\{
    \begin{array}{lrr}
      \textbf{q}(t)  & \mbox{ with probability } & r(t)\\
      \textbf{p}(t) & \mbox{ with probability } & 1 - r(t)
    \end{array}
  \right.
\end{align}
\item If $t<T$, return to step 1 and increment $t$.  Otherwise, stop.
\end{enumerate}

After running the above algorithm, \cite{RigTsy10} then approximate the sparsity pattern aggregate as:
\begin{align}
\widehat{\boldsymbol{\beta}}^{SPA} = \frac{1}{T-T_0}\sum_{t = T_0}^{T} \widehat{\boldsymbol{\beta}}_{\textbf{p}(t)}.
\end{align}
Here, $T_0$ is an arbitrary integer, used to allow for convergence of
the Markov chain. Note that the above algorithm can be applied to any prior for the class of aggregation estimators considered in this paper, we need only update Equation~\ref{priorRatio}.

In the above algorithm, $\sigma^2$ was assumed known. In general applications, $\sigma^2$ is unknown. \cite{RigTsy10} gave
the following strategy for dealing with this case. Denote
$\widehat{\boldsymbol{\beta}}^{SPA}_{\delta}$ as the sparsity pattern estimator
computed with $\sigma^2 = \delta$. Then, we estimate $\sigma^2$ as:
\begin{align}
\widehat{\sigma}^2 = \inf \left\{ \delta: \left| \frac{ || \textbf{y}
      - \textbf{X}\widehat{\boldsymbol{\beta}}^{SPA}_{\delta}||_2^2  }{n -
      M_n(\widehat{\boldsymbol{\beta}}^{SPA}_{\delta})} - \delta \right| > \alpha \right\},
\end{align}
where $\alpha >0$ is a tolerance parameter, and $M_n(\boldsymbol{\beta}) =
\sum_{j=1}^M \textbf{1}_{|\boldsymbol{\beta}_j| > 1/n}$.  Again, this strategy
needs no modification if the prior is changed. 

Note that while the sparse aggregation
estimator in up-weights sparse models via
the prior, it does not exclude any models.  The exact estimator is
therefore not sparse.  However, this computational strategy nearly
always results in a sparse estimate.  This is because the Markov chain simply does not visit any models that are not sparse. Similarly, while the structured
sparse priors we introduce do not
eliminate structured sparse models from the exact aggregate estimators, the computed estimators almost
always have this property. Alternately, we could run the Markov chain for a very long time, and obtain a model that includes all covariates. However, we would see that many covariates appear very seldom in the chain, and we could thus obtain a sparse or structured sparse solution with a simple thresholding strategy.

\subsection{Structural Modifications of the Algorithm}

In structured sparse aggregation, we can take advantage of the allowed sparsity patterns in the prior to streamline the metropolis algorithm. For grouped sparsity, we instead consider the hypercube of groups instead of the hypercube of all predictors. That is, given a set of groups $\mathcal{G}$, we instead consider patterns represented by $\{0,1\}^{|\mathcal{G}|}$. Effectively, we consider adding and removing groups as a whole, rather than individual coordinates. In the case of strong hierarchical sparsity, we can exclude any neighboring patterns that do not satisfy strong hierarchy. That is, given a DAG, we only consider adding direct descendants of the current sparsity pattern, or removing leaf nodes with respect to the current sparsity pattern.

\section{Proof and Theoretical Framework}
\label{appendix}

In the following sections, we give a general theoretical recipe, leading to the results in Section 4 in the main text. In Section~\ref{mauryLemma}, we give two Lemmas for our specific applications.

\subsection{Priors and Set Function Bounds}

\begin{lem}
\label{general31lemma1}
(From~\cite{RigTsy10}) Fix $\textbf{p} \in \{0,1\}^M$, assume that $\xi_i$ are iid random
variables such that $\mathbb{E}\xi_i = 0$, and $\mathbb{E}\xi_i^2 =
\sigma^2$, for $i=1,\ldots,n$.  Then for least squares estimator:
\begin{align}
\label{OLS}
\widehat{\boldsymbol{\beta}}_\textbf{p} &= \underset{\boldsymbol{\beta} \in \mathbb{R}^M: \ 
  \mbox{supp}(\boldsymbol{\beta}) \subseteq \mbox{supp}(\textbf{p})}{\operatorname{argmin}}\| \textbf{y} - \textbf{X}\boldsymbol{\beta}\|^2_2,
\end{align}
we have:
\begin{align}
\mathbb{E} || \textbf{X}{\widehat{\boldsymbol{\beta}}_{\textbf{p}}} - \textbf{y} ||^2_2 &\leq
\min_{\boldsymbol{\beta} \in
  \mathbb{R}^\textbf{p}} || \textbf{X}{\boldsymbol{\beta}} - \textbf{y} ||^2_2 +
\sigma^2\frac{\min(||\textbf{p}||_0,R)}{n}.
\end{align}
Where $R = \mbox{rank}(\textbf{X})$.
\end{lem}
Now, suppose that  we have a set function $\mathcal{M}: 2^\mathcal{I}
\to \mathbb{R}^+$. We then define $M_\mathcal{M}(x): \mathbb{R}^M \to
\mathbb{R}^+$ as $\mathcal{M}(\mbox{supp}(x))$. We then use a
prior of the form: 

\begin{align}
\label{generalPrior}
\pi_{\textbf{p},\mathcal{M}} := \left\{
     \begin{array}{lr}
       \frac{1}{H_\mathcal{M}} \left( \frac{M_\mathcal{M}(\textbf{p})}{2eC}
       \right)^{M_\mathcal{M}(\textbf{p})} & ||\textbf{p}||_0 \leq R\\
       \frac{1}{2} & ||\textbf{p}||_0 = M\\
       0 & \mbox{else}
     \end{array}
   \right. .
 \end{align}
Here $R$ is the rank of $\textbf{X}$ and $C\geq1$ is such that the normalizing constant $H_\mathcal{M}
\leq 4$.  Note that $4$ is an arbitrary constant used for the sake of
consistency throughout the theory presented here and in~\cite{RigTsy10}. 

\begin{lem}
\label{general31lemma2}
(From~\cite{leungInfoMixing},~\cite{RigTsy10}) 
Consider the sparsity pattern
estimator with prior $\pi_{\mathcal{M}}$: $\textbf{X}{\widehat{\boldsymbol{\beta}}^{\mathcal{M}}}$, then:
\begin{align}
\mathbb{E}|| \textbf{X}{\widehat{\boldsymbol{\beta}}^{\mathcal{M}}} - \textbf{y} ||^2_2 \leq \min_{\textbf{p}
  \in \{0,1\}^M; \pi_{\textbf{p},\mathcal{M}} \neq 0} \left\{ \mathbb{E} ||
  \textbf{X}{\widehat{\boldsymbol{\beta}}_\textbf{p}} - \textbf{y} ||^2_2 + \frac{4\sigma^2 \log(
    \pi_{\textbf{p},\mathcal{M}}^{-1})}{n} \right\}
\end{align}
\end{lem}
We now make the
following assumption:
\begin{asm}
\label{generalAssumptionPi}
For all $\textbf{p} \in \mathcal{P}$ where $R>||\textbf{p}||_0>0$:
\begin{align}
\frac{||\textbf{p}||_0}{M_\mathcal{M}(\textbf{p})} \leq
\log\left( 1 + \frac{eC}{\max(M_\mathcal{M}(\textbf{p}),1)}\right).
\end{align}
\end{asm}
Now, for $\textbf{p}$ such that $||\textbf{p}||_0 < R$, the following holds:
\begin{lem}
\label{general31lemma3}
\begin{align}
\frac{4\sigma^2 \log( \pi_{\textbf{p},\mathcal{M}}^{-1})}{n} \leq \frac{8\sigma^2
  M_\mathcal{M}(\textbf{p})}{n} \log \left(1 +
  \frac{eM}{\max(M_\mathcal{M}(\textbf{p}),1) } \right) + \frac{8
  \sigma^2}{n} \log2.
\end{align}
\end{lem}

We now present the main general result:

\begin{prp} 
\label{general31}
For any $M\geq1, n \geq1$, the sparsity pattern
estimator with prior $\pi_{\mathcal{M}}$: $\textbf{X}{\widehat{\boldsymbol{\beta}}^{\mathcal{M}}}$ satisfies:
\begin{align}
\mathbb{E} ||\textbf{X}{\widehat{\boldsymbol{\beta}}^{\mathcal{M}}} - \textbf{y} ||^2_2 \leq 
\min_{\boldsymbol{\beta} \in \mathbb{R}^M} \left\{ ||\textbf{X}{\boldsymbol{\beta}} - \textbf{y}||^2_2 + \min
  \left\{\frac{\sigma^2 R}{n}, \mbox{ } 9 \sigma^2
  \frac{M_\mathcal{M}(\boldsymbol{\beta})}{n} \log \left( 1 + 
    \frac{eC}{\max(M_\mathcal{M}(\boldsymbol{\beta}), 1)}\right) \right\} \right\} +
\frac{8\sigma^2}{n}\log 2.
\end{align}
\end{prp}

\begin{proof}
For $||\boldsymbol{\beta}||_0 \leq R$, we know from combining Lemma~\ref{general31lemma2},
Lemma~\ref{general31lemma3}, and Assumption~\ref{generalAssumptionPi} that
$\textbf{X}{\widehat{\boldsymbol{\beta}}^{\mathcal{M}}}$ satisfies: 
\begin{align}
\mathbb{E} ||\textbf{X}{\widehat{\boldsymbol{\beta}}^{\mathcal{M}}} - \textbf{y} ||^2_2 \leq 
\min_{\boldsymbol{\beta} \in \mathbb{R}^M, ||\boldsymbol{\beta}||_0 < R} \left\{ ||
  \textbf{X}{\boldsymbol{\beta}} - \textbf{y} ||^2_2 + \frac{9\sigma^2
  M_\mathcal{M}(\boldsymbol{\beta})}{n} \log \left( 1+
  \frac{eC}{\max(M_\mathcal{M}(\boldsymbol{\beta}),1) } \right) \right\} + \frac{8
  \sigma^2}{n} \log2 .
\end{align}
For $||\boldsymbol{\beta}||_0 = M$, we have:
\begin{align}
\mathbb{E} ||\textbf{X}{\widehat{\boldsymbol{\beta}}^{\mathcal{M}}} - \textbf{y} ||^2_2 \leq \min_{\boldsymbol{\beta} \in \mathbb{R}^M} \left\{  ||
  \textbf{X}{\boldsymbol{\beta}} - \textbf{y} ||^2_2 + \sigma^2\frac{R}{n} \right\} + \frac{4
  \sigma^2}{n} \log2 .
\end{align}
And so the proposition follows directly.
\end{proof}

\subsection{Convex Norm Bounds}

\begin{lem}
\label{lemma82genHelp}
For integer $M>0$, define $\mathcal{I} = \{1,\ldots,M\}$. Suppose that
we have a set function $\mathcal{M}: 2^\mathcal{I} \to \mathbb{R}^+$ and
norm $||\cdot||_\mathcal{M}: \mathbb{R}^M \to \mathbb{R}$. We then
define $M_\mathcal{M}(x): \mathbb{R}^M \to \mathbb{R}^+$ as
$\mathcal{M}(\mbox{supp}(x))$.  Then, if for any $\boldsymbol{\beta}^* \in
\mathbb{R}^M \backslash \{0\}$, any integer $k\geq 1$, and any
function $f$ we have:
\begin{align}
\min_{\boldsymbol{\beta}: ||\boldsymbol{\beta}||_\mathcal{M} = ||\boldsymbol{\beta}^*||_\mathcal{M}; M_\mathcal{M}(\boldsymbol{\beta}) \leq k} ||f -
\textbf{X}{\boldsymbol{\beta}}||^2 \leq ||f - \textbf{X}{\boldsymbol{\beta}^*}||^2 + \frac{||\boldsymbol{\beta}^*||_\mathcal{M}^2}{\min(k,M_\mathcal{M}(\boldsymbol{\beta}^*))}.
\end{align}

Then, for any $C\geq 1$, integer $n>0$, constant $\nu>0$, a given real
number $k^*\geq1$, $\boldsymbol{\beta}^* \in \mathbb{R}^M \backslash \{0\}$, and
any function $\overline{M}_\mathcal{M}(x)$ that satisfies, for some $\gamma > 0$; $M_\mathcal{M}(x) \leq \overline{M}_\mathcal{M}(x) \leq \gamma M_\mathcal{M}(x) \ \forall
x \in \mathbb{R}^M$:
\begin{align}
\min_{\boldsymbol{\beta} \in \mathbb{R}^M} \left\{ ||\textbf{X}{\boldsymbol{\beta}} - \textbf{y}||^2_2 +
  \nu^2\frac{\overline{M}_\mathcal{M}(\boldsymbol{\beta})}{n}\log\left( 1 +
    \frac{eC}{\max(\overline{M}_\mathcal{M}(\boldsymbol{\beta}),1)}\right) \right\} \leq \notag \\
    \leq ||\textbf{X}{\boldsymbol{\beta}^*}
- \textbf{y}||^2_2 + \nu^2\frac{k^*}{n}\log\left( 1 + \frac{eC}{k^*} \right) +
\frac{\gamma ||\boldsymbol{\beta}^*||_\mathcal{M}^2}{k^*}.
\end{align}
\end{lem}

\begin{proof}
We consider two cases, $k^* \leq \overline{M}_\mathcal{M}(\boldsymbol{\beta}^*)$, and $k^* >
\overline{M}_\mathcal{M}(\boldsymbol{\beta}^*)$

\begin{itemize}
\item Let $k^* \leq \overline{M}_\mathcal{M}(\boldsymbol{\beta}^*)$.
\begin{align}
 & \min_{\boldsymbol{\beta} \in \mathbb{R}^M} \left\{ ||\textbf{X}{\boldsymbol{\beta}} - \textbf{y}||^2_2 +
  \nu^2\frac{\overline{M}_\mathcal{M}(\boldsymbol{\beta})}{n}\log\left( 1 +
    \frac{eC}{\max(\overline{M}_\mathcal{M}(\boldsymbol{\beta}),1)}\right) \right\} \leq \notag \\
    \leq & \min_{1\leq k \leq \overline{M}_\mathcal{M}(\boldsymbol{\beta}^*)}
\min_{\boldsymbol{\beta} \in \mathbb{R}^M:\overline{M}_\mathcal{M}(\boldsymbol{\beta}) \leq k} \left\{ ||\textbf{X}{\boldsymbol{\beta}} - \textbf{y}||^2_2 +
  \nu^2\frac{\overline{M}_\mathcal{M}(\boldsymbol{\beta})}{n}\log\left( 1 +
    \frac{eC}{\max(\overline{M}_\mathcal{M}(\boldsymbol{\beta}),1)}\right) \right\}   \\
\leq & \min_{1\leq k \leq \overline{M}_\mathcal{M}(\boldsymbol{\beta}^*)}
\min_{\boldsymbol{\beta} \in \mathbb{R}^M:M_\mathcal{M}(\boldsymbol{\beta}) \leq k/\gamma} \left\{ ||\textbf{X}{\boldsymbol{\beta}} - \textbf{y}||^2_2 +
  \nu^2\frac{\gamma M_\mathcal{M}(\boldsymbol{\beta})}{n}\log\left( 1 +
    \frac{eC}{\max(\gamma M_\mathcal{M}(\boldsymbol{\beta}),1)}\right) \right\}\\
\leq & \min_{1\leq k \leq \overline{M}_\mathcal{M}(\boldsymbol{\beta}^*)} \left\{
\min_{\boldsymbol{\beta} \in \mathbb{R}^M: ||\boldsymbol{\beta}||_\mathcal{M} = ||\boldsymbol{\beta}^*||_\mathcal{M};  M_\mathcal{M}(\boldsymbol{\beta}) \leq k/\gamma} \left\{ ||\textbf{X}{\boldsymbol{\beta}} - \textbf{y}||^2_2 \right\} +
  \nu^2\frac{k}{n}\log\left( 1 +
    \frac{eC}{k}\right)  \right\}\\
\leq & \ ||\textbf{X}{\boldsymbol{\beta}^*} - \textbf{y} ||^2_2 + \min_{1\leq k \leq
  \overline{M}_\mathcal{M}(\boldsymbol{\beta}^*)} \left\{ \frac{\gamma ||\boldsymbol{\beta}^*||_\mathcal{M}^2}{k}+
  \nu^2\frac{k}{n}\log\left( 1 +
    \frac{eC}{k}\right) \right\}\\
\leq & \ ||\textbf{X}{\boldsymbol{\beta}^*}
- \textbf{y}||^2_2 + \nu^2\frac{k^*}{n}\log\left( 1 + \frac{eC}{k^*} \right) +
\frac{\gamma||\boldsymbol{\beta}^*||_\mathcal{M}^2}{k^*} 
\end{align}


In the above, we use the monotonicity of the mapping $g(t) =
\frac{t}{n}\log\left( 1 + \frac{eC}{t}\right)$ for $t\geq1$. We apply the assumptions on $M_\mathcal{M}$ in the next steps.  We finally use the fact that $k^* \leq
\overline{M}_\mathcal{M}(\boldsymbol{\beta}^*)$ in the fourth step.  This
completes this case.  

\item For $k^* > \overline{M}_\mathcal{M}(\boldsymbol{\beta}^*)$, we can use a simple argument:
\begin{align}
& \ \min_{\boldsymbol{\beta} \in \mathbb{R}^M} \left\{ ||\textbf{X}{\boldsymbol{\beta}} - \textbf{y}||^2_2 +
  \nu^2\frac{\overline{M}_\mathcal{M}(\boldsymbol{\beta})}{n}\log\left( 1 +
    \frac{eC}{\max(\overline{M}_\mathcal{M}(\boldsymbol{\beta}),1)}\right) \right\} \leq \notag \\
\leq & \ ||\textbf{X}{\boldsymbol{\beta}^*} - \textbf{y}||^2_2 +
  \nu^2\frac{\overline{M}_\mathcal{M}(\boldsymbol{\beta}^*)}{n}\log\left( 1 +
    \frac{eC}{\max(\overline{M}_\mathcal{M}(\boldsymbol{\beta}^*),1)}\right) \\
\leq & \ ||\textbf{X}{\boldsymbol{\beta}^*} - \textbf{y}||^2_2 +
  \nu^2\frac{k^*}{n}\log\left( 1 +
    \frac{eC}{k^*}\right) +
\frac{\gamma ||\boldsymbol{\beta}^*||_\mathcal{M}^2}{k^*}
\end{align}
\end{itemize}
These two cases complete the proof.
\end{proof}

Note that Lemmas~\ref{lemma81penalizedL0} and~\ref{lemma81groupedL0}
give results that guarantee that the conditions of the above lemma are
satisfied in the two important cases considered in this paper. We may
also use $M_\mathcal{M}(\cdot)$ with $\gamma = 1$ in place of $\overline{M}_\mathcal{M}$
at all points in this lemma and obtain the same result in terms of
$M_\mathcal{M}(\cdot)$.  In light of this result, we now give a
generalized version of Lemma 8.2 from~\cite{RigTsy10}.  The result of the
lemma has been simplified from the version in~\cite{RigTsy10}, but the full
result still holds. 

\begin{prp}
\label{general82}
Assume all of the conditions of Lemma~\ref{lemma82genHelp}.  Then,
\begin{align}
 \mathbb{E}||\textbf{X}{\widehat{\boldsymbol{\beta}}^{\mathcal{M}}} - \textbf{y} ||^2_2 \leq \min_{\boldsymbol{\beta} \in
   \mathbb{R}^M} \{ ||\textbf{X}{\boldsymbol{\beta}} - \textbf{y}||^2_2 + \phi_{n,M}(\boldsymbol{\beta}) \} +
 \frac{\sigma^2}{n}(9\log(1+eM) + 8\log2)
\end{align}
where $\phi_{n,M,C,\mathcal{M}}(0) := 0$ and for $\boldsymbol{\beta} \neq 0$:
\begin{align}
\phi_{n,M,C,\mathcal{M}} = \min\left[ \frac{\sigma^2}{n},\frac{9\sigma^2
    \overline{M}_\mathcal{M}(\boldsymbol{\beta})}{n} \log \left( 1+
    \frac{eC}{\max(\overline{M}_\mathcal{M}(\boldsymbol{\beta}),1)}\right), \frac{11\sigma
    \sqrt{\gamma}||\boldsymbol{\beta}||_\mathcal{M}}{\sqrt{n}}\sqrt{\log \left( 1 +
      \frac{3eC\sigma}{||\boldsymbol{\beta}||_\mathcal{M}\sqrt{\gamma n}}\right)}  \ \right]
\end{align}

\begin{proof}
We first show the following:
\begin{align}
\min_{\boldsymbol{\beta} \in \mathbb{R}^M} \left\{ ||\textbf{X}{\boldsymbol{\beta}} - \textbf{y}||^2_2 +
  \nu^2\frac{\overline{M}_\mathcal{M}(\boldsymbol{\beta})}{n}\log\left( 1 +
    \frac{eC}{\max(\overline{M}_\mathcal{M}(\boldsymbol{\beta}),1)}\right)
\right\} \leq  \notag \\
\leq \min_{\boldsymbol{\beta} \in \mathbb{R}^M} \left\{ ||\textbf{X}{\boldsymbol{\beta}} - \textbf{y} ||_2^2 + (3 +
1/e) \overline{\phi}_{n,M,C,\mathcal{M}}(\boldsymbol{\beta}) \right\}.\label{subGen82}
\end{align}
Where $\overline{\phi}_{n,M,C,\mathcal{M}}(\boldsymbol{\beta}) = 0$ for $\boldsymbol{\beta}
=0$ and otherwise:
\begin{align}
\overline{\phi}_{n,M,C,\mathcal{M}}(\boldsymbol{\beta}) =  \frac{\nu\sqrt{\gamma}||\boldsymbol{\beta}||_\mathcal{M}}{\sqrt{n}} \sqrt{\log
    \left( 1 + \frac{eC\nu}{||\boldsymbol{\beta}||_\mathcal{M} \sqrt{\gamma n}}\right)}
  + \frac{\nu^2\log(1 + eC)}{(3 + 1/e)n}.
\end{align}

It is clear that Equation~\ref{subGen82} holds for $\boldsymbol{\beta} = 0$.  For $\boldsymbol{\beta} \neq 0$, we
begin with the statement of Lemma~\ref{lemma82genHelp}: 
\begin{align}
\min_{\boldsymbol{\beta} \in \mathbb{R}^M} \left\{ ||\textbf{X}{\boldsymbol{\beta}} - \textbf{y}||^2_2 +
  \nu^2\frac{\overline{M}_\mathcal{M}(\boldsymbol{\beta})}{n}\log\left( 1 +
    \frac{eC}{\max(\overline{M}_\mathcal{M}(\boldsymbol{\beta}),1)}\right) \right\} \leq \\
    \leq ||\textbf{X}{\boldsymbol{\beta}^*}
- \textbf{y}||^2_2 + \nu^2\frac{k^*}{n}\log\left( 1 + \frac{eC}{k^*} \right) +
\frac{\gamma ||\boldsymbol{\beta}^*||_\mathcal{M}^2}{k^*}
\end{align}
Then, using the proof of Lemma 8.2 in~\cite{RigTsy10}, we can show:
\begin{align}
\nu^2\frac{k^*}{n}\log\left( 1 + \frac{eC}{k^*} \right) +
\frac{\gamma||\boldsymbol{\beta}^*||_\mathcal{M}^2}{k^*} \leq (3 + 1/e) \overline{\phi}_{n,M,C,\mathcal{M}}(\boldsymbol{\beta}^*).
\end{align}
We next let $\sigma = \nu$ and combine Equation~\ref{subGen82} with Lemma~\ref{general31} to complete the proof. The constants are finally rounded up to the nearest integer for clarity as in~\cite{RigTsy10}.
\end{proof}
\end{prp}

\subsection{Lemmas for Specific Norms and Set Functions}
\label{mauryLemma}

We now give a two lemmas guaranteeing that the conditions in
Lemma~\ref{lemma82genHelp} are satisfied in two particular settings.  The
following lemma is a notationally adapted version of lemma 8.1
in~\cite{RigTsy10} and is given without proof:  
\begin{lem}
\label{lemma81penalizedL0}
For any $\boldsymbol{\beta}^* \in \mathbb{R}^M \backslash \{0\}$, any integer
$k\geq 1$, $\textbf{X}$ such that $\max_{1 \leq j \leq M}
||\textbf{x}_j||_2 \leq 1$, and any vector $\textbf{y}$:
\begin{align}
\min_{\boldsymbol{\beta}: |\boldsymbol{\beta}|_1 = |\boldsymbol{\beta}^*|_1; ||\boldsymbol{\beta}||_0 \leq k} ||\textbf{y} -
\textbf{X}{\boldsymbol{\beta}}||^2 \leq ||\textbf{y} - \textbf{X}{\boldsymbol{\beta}^*}||^2 + \frac{||\boldsymbol{\beta}^*||_1^2}{\min(k,||\boldsymbol{\beta}^*||_0)}
\end{align}
\end{lem}

We now give a version of this lemma for grouped $\ell_0$-like norms:

\begin{lem}
\label{lemma81groupedL0}
For any $\boldsymbol{\beta}^* \in \mathbb{R}^M \backslash \{0\}$, any integer
$k\geq 1$, $\textbf{X}$ such that $\max_{1 \leq j \leq M}
||\textbf{x}_j||_2 \leq 1$, and any vector $\textbf{y}$:
\begin{align}
\min_{\boldsymbol{\beta}: |\boldsymbol{\beta}|_{1,\mathcal{G}} = |\boldsymbol{\beta}^*|_{1,\mathcal{G}}; ||\boldsymbol{\beta}||_{0,\mathcal{G}} \leq k} ||\textbf{y} -
\textbf{X}{\boldsymbol{\beta}}||^2 \leq ||\textbf{y} - \textbf{X}{\boldsymbol{\beta}^*}||^2 + \frac{||\boldsymbol{\beta}^*||_{1,\mathcal{G}}^2}{\min(k,||\boldsymbol{\beta}^*||_{0,\mathcal{G}})}
\end{align}
\end{lem}
\begin{proof}
Fix $\boldsymbol{\beta}^* \in \mathbb{R}^M \backslash \{0\}$, and integer $k\geq 1$.  Set
$K = \min(k, ||\boldsymbol{\beta}^*||_{0,\mathcal{G}})$.  Let
$\mathcal{V}_\mathcal{G}(\boldsymbol{\beta}^*) = \{v_g^*\}$ be a
$\mathcal{G}$-decomposition of $\boldsymbol{\beta}^*$ minimizing the norm
$||\boldsymbol{\beta}^*||_{1,\mathcal{G}}$.  Then, define the multinomial
parameter, a $|\mathcal{G}|$ vector, $\textbf{q} = \{q_1,\ldots,
q_{|\mathcal{G}|}\}$, with $q_g =
\frac{||v_g^*||_2}{||\boldsymbol{\beta}^*||_{1,\mathcal{G}}}$.  Let the
$|\mathcal{G}|$-vector $\kappa$ have multinomial distribution
$\mathbb{M}(K,\textbf{q})$.  We then define the random vector
$\widetilde{\boldsymbol{\beta}} \in \mathbb{R}^{\sum_\mathcal{G} |g|}$, a
concatenation of $|\mathcal{G}|$ vectors: $[\widetilde{\boldsymbol{\beta}}_g]_{g \in
  \mathcal{G}}$, with components $\widetilde{\boldsymbol{\beta}}_g = \frac{\kappa_g
  v_g^*||\boldsymbol{\beta}^*||_{1,\mathcal{G}}}{K ||v^*_g||_2}$.  Here, we
adopt the convention that if
$v_g^*/||v_g^*||_2 = 0/0$, then $v_g^*/||v_g^*||_2 = 0$ (note that
$q_g =0$ in this case).  Thus, we have
that $\mathbb{E}\widetilde{\boldsymbol{\beta}}_g = v_g^*$, and $\mathbb{V}(\kappa_g) =
Kq_g(1-q_g)$.  Now, we have that the entries on the diagonal of the
covariance matrix of $\widetilde{\boldsymbol{\beta}}$ is bounded as follows: 
\begin{align}
\mbox{Diag}(\widetilde{\Sigma}_g) \leq \frac{||\boldsymbol{\beta}^*||_{1,\mathcal{G}}||v_g^*||_2}{K}\textbf{1}_{|g|}.
\end{align}
Here $\textbf{1}_{|g|}$ is a $|g|$ length vector with entries all
equal to $1$.  Now, let $\overline{\boldsymbol{\beta}} \in
\mathbb{R}^M$ be such that
$\mathcal{V}_\mathcal{G}(\overline{\boldsymbol{\beta}}) = \{
v_g(\widetilde{\boldsymbol{\beta}}) \}$, where
$v_g(\widetilde{\boldsymbol{\beta}}) \in \mathbb{R}^M$ is equal to
$\widetilde{\boldsymbol{\beta}}_g$ for the indices in $g$, and equal
to zero otherwise. 

Then, we define the following $n \times \sum_{g \in
  \mathcal{G}}|g|$ matrix: $\widetilde{\textbf{X}} = [\textbf{x}_j: j \in g]_{g
  \in \mathcal{G}}$, where $\textbf{x}_i$ is the $i$th column of
$\textbf{X}$. Let $\widetilde{\boldsymbol{\beta}^*} =
[\boldsymbol{\beta}_g^*]_{g\in\mathcal{G}}$. Now, if $\max_{i \in \mathcal{I}} ||\textbf{x}_i||_2 \leq 1$, then for any vector $\textbf{y}$ we have:
\begin{align}
\mathbb{E}||\textbf{y} - \textbf{X}\overline{\boldsymbol{\beta}}||_2^2 &= \mathbb{E}||\textbf{y} - \widetilde{\textbf{X}}{\widetilde{\boldsymbol{\beta}}}||_2^2\\
 &= ||\textbf{y} - \widetilde{\textbf{X}}\widetilde{{\boldsymbol{\beta}^*}}||_2^2 + \frac{1}{n}\sum_{i=1}^n \widetilde{\textbf{x}}_i^T\widetilde{\Sigma} \widetilde{\textbf{x}}_i\\
&\leq ||\textbf{y} - \textbf{X}{\boldsymbol{\beta}^*}||_2^2 + \frac{||\boldsymbol{\beta}^*||_{1,\mathcal{G}}^2}{K}
\end{align}
In the above $\widetilde{\textbf{x}}_i$ is the $i$th row of $\widetilde{\textbf{X}}$.
It is clear that $||\overline{\boldsymbol{\beta}}||_{0,\mathcal{G}}
\leq K$. Further, by Corollary 1 from~\cite{Jacob:2009:GLO:1553374.1553431}, $||\overline{\boldsymbol{\beta}}||_{1,\mathcal{G}} = ||\boldsymbol{\beta}^*||_{1,\mathcal{G}}$.  The result then follows.
\end{proof}

\subsection{Discussion}

 The theoretical framework presented here  leads
us to postulate that there are many potential aggregate estimators
that give similar theoretical guarantees. However, the assumptions in
Lemma~\ref{lemma82genHelp} play a key role. Beginning with a set
function and its corresponding convex extension, we
could propose a prior that would give us an aggregate estimator that
would enjoy adaptation to patterns in terms of both the set function
and the convex norm. In light of the assumptions in
Lemma~\ref{lemma82genHelp}, it is necessary to produce a general form
of Lemmas~\ref{lemma81penalizedL0} and~\ref{lemma81groupedL0} --- which give a bound on the approximation of $\textbf{y}$ when we restrict the approximating functions to a class that depends on our set function.  Such a
result remains an open question, but we suspect it is not attainable for
all set functions. However, since such a result needs to hold only for
a set function (and its corresponding convex norm) that is a
bounded between our target function, there may exists several interesting
extensions that we have not proposed in this paper.


\begin{table}

\begin{center}
\caption{Simulation results for 1-dimensional linear sparsity
  patterns, see e.g. Figure~\ref{oneDimOneRunFigure}. Both sparse
  (SPA) and structured sparse (SSA) aggregation methods outperform the
  lasso in terms of prediction and recovery of the true sparsity
  pattern.  Note that for paired runs of both aggregation methods on a
  single simulated data set, the structured estimator is superior in
  both measures at least 95\% of the time. For each measure, the mean
  over 250 trials is reported with the standard error in parentheses.}
  \label{penalizedEllZeroSimulationOneD}
\begin{tabular}{ | c  | c c c | }
\hline
  ($n$, $M$, $C$, $C_{on}$, $\sigma$) & Prediction (SPA)&
  Prediction (SSA) & Prediction (lasso) \\ 
\hline
(100, 100, 1, 9, 1) & 0.168 (0.104) & 0.123 (0.065) & 0.813 (0.347) \\
(200, 500, 1, 20, 1.5) & 0.371 (0.132) & 0.298 (0.156) & 2.765 (0.723) \\
(100, 100, 2, 5, 1.1) & 0.156 (0.078) & 0.138 (0.079) & 1.012 (0.46)
\\
\hline
\hline
  ($n$, $M$, $C$, $C_{on}$, $\sigma$) & Recovery (SPA)&
  Recovery (SSA) &Recovery (lasso)\\ 
(100, 100, 1, 9, 1) & 0.018 (0.01) & 0.014 (0.006) & 0.089 (0.035) \\
(200, 500, 1, 20, 1.5) & 0.019 (0.007) & 0.015 (0.007) & 0.14 (0.034) \\
(100, 100, 2, 5, 1.1) & 0.016 (0.007) & 0.014 (0.007) & 0.104 (0.044) \\
\hline
\end{tabular}
\end{center}

\end{table}

\begin{table}

\begin{center}
\caption{Simulation results for 2-dimensional lattice sparsity
  patterns, see e.g. Figure~\ref{twoDimOneRunFigure}. Both sparse
  (SPA) and structured sparse (SSA) aggregation methods outperform the
  lasso in terms of prediction and recovery of the true sparsity
  pattern.  Note that for paired runs of both aggregation methods on a
  single simulated data set, the structured estimator is superior in
  both measures at least 95\% of the time.  For each measure, the mean
  over 250 trials is reported with the standard error in parentheses. }
  \label{penalizedEllZeroSimulationTwoD}
\begin{tabular}{ | c  | c c c | }
\hline
  ($n$, $M$, $C$, $C_{on}$, $\sigma$) & Prediction (SPA)&
  Prediction (SSA) & Prediction (lasso)\\ 
\hline
(100, 100, 1, 9, 1) & 0.131 (0.063) & 0.113 (0.059) & 0.844 (0.398) \\
(200, 400, 2, 9, 1.4) & 0.265 (0.088) & 0.214 (0.085) & 2.061 (0.616) \\
\hline
\hline
  ($n$, $M$, $C$, $C_{on}$, $\sigma$) & Recovery (SPA)&
  Recovery (SSA) & Recovery (lasso)\\ 
\hline
(100, 100, 1, 9, 1) & 0.015 (0.006) & 0.013 (0.006) & 0.094 (0.043) \\
(200, 400, 2, 9, 1.4) & 0.015 (0.004) & 0.012 (0.004) & 0.115 (0.032) \\
\hline
\end{tabular}
\end{center}

\end{table}

\begin{table}
\begin{center}
\caption{Comparison of predictive power for the HIV data. We estimated
  the testing errors using three fold data splitting.  We see that
  the structured aggregation estimator gives comparable predictive
  performance to the sparse aggregation estimator. The structured
  estimator also carries the benefit of better biological
  interpretability. The mean test error is given with the standard
  errors in parentheses.}
  \label{hivTable}
\begin{tabular}{ | c  | c c c c | }
\hline
Data Splitting  & Sparse  & Structured  & Stepwise & \\ 
Test Error &  Aggregation &  Aggregation & Regression & lasso \\ 
\hline
 \rule[-.3cm]{0cm}{.8cm} $\left|\left|\textbf{y}_{\mbox{test}} -
 \textbf{X}_{\mbox{test}}\widehat{\boldsymbol{\beta}}\right|\right|_2^2$ &
0.65 (0.04)& 0.69 (0.07) & 3.03 (0.18)& 1.45 (0.2)\\
\hline
\end{tabular}
\end{center}

\end{table}


\begin{figure}[t]
\centering
\includegraphics[width = 5in]{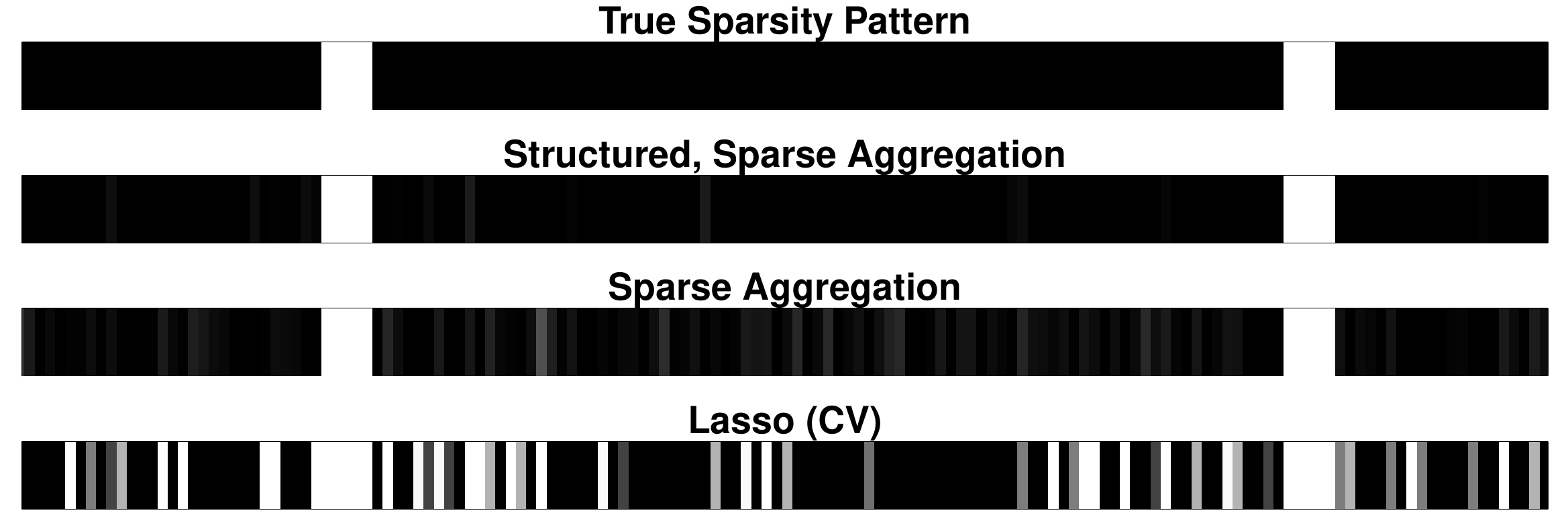}
\caption{A structured sparsity setting for a linear structure in the coefficient vector --- adjacent entries of $\boldsymbol{\beta}$ are considered close. We display the true linear sparsity pattern (top) and recovered sparsity patterns
  by structured, sparse aggregation (second from top), sparse
  aggregation (third from top), and cross-validated lasso
  (bottom). Black (0\%) to white (100\%) indicates percentage of
  selection in the Markov Chain algorithm for the aggregation estimators.  For the lasso, black (0\%) to white (100\%) indicates the percentage of selection out of 100 replications of cross validation.
  Structured sparse aggregation is able to best recover the true sparsity pattern. Both sparse aggregation and the lasso suffer from false positives scattered throughout the space of candidate covariates.}
\label{oneDimOneRunFigure}
\end{figure}

\begin{figure}[t]
\centering
\includegraphics[width = 5in]{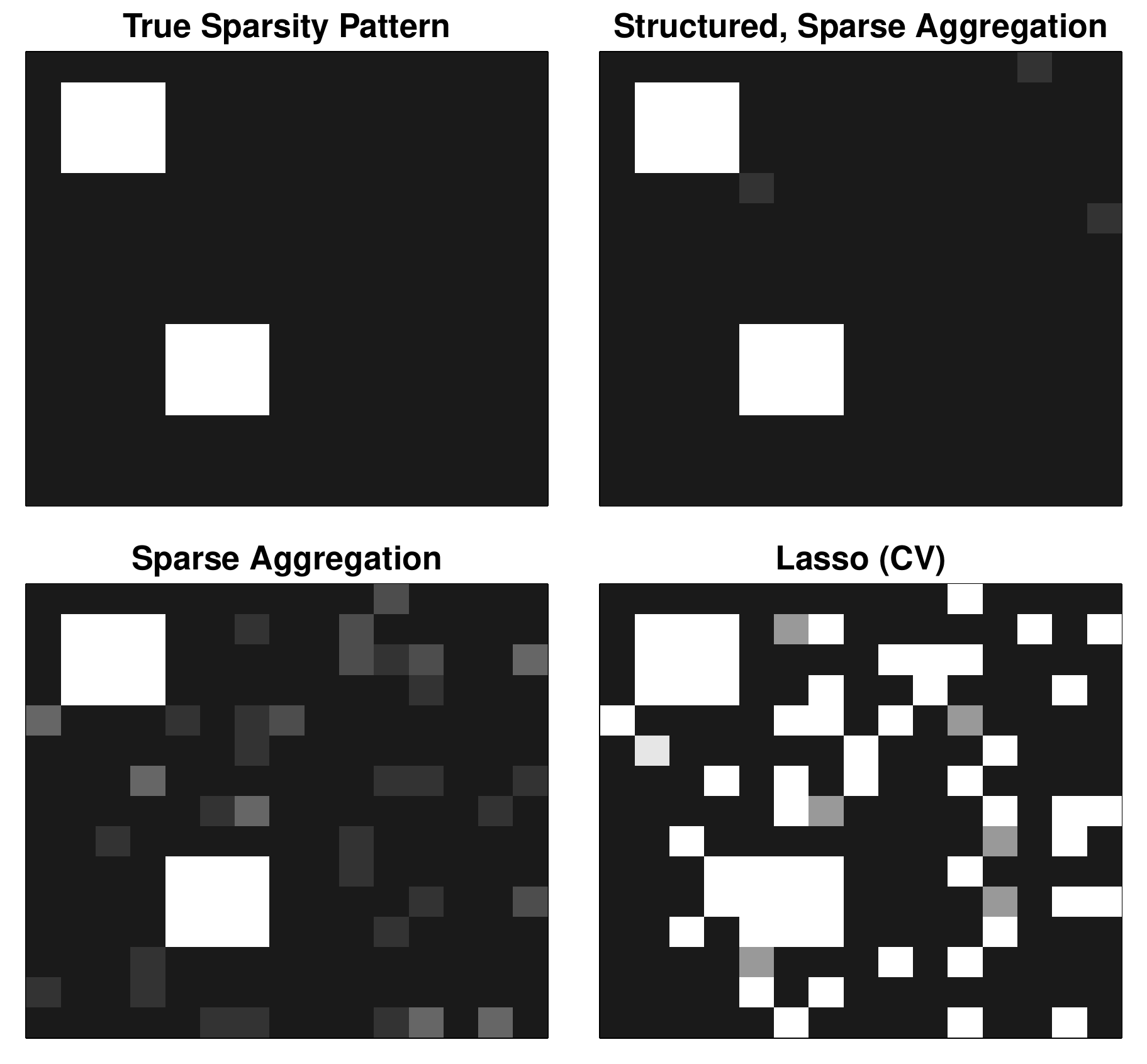}
\caption{An example of structured sparsity in a two dimensional lattice --- the coefficient vector $\boldsymbol{\beta}$ is an unraveled matrix. We display the true 2-D lattice sparsity pattern (top left) and the recovered sparsity patterns
  by structured, sparse aggregation (top right), sparse
  aggregation (bottom left), and cross-validated lasso
  (bottom right). Black (0\%) to white (100\%) indicates percentage of
  selection in the Markov Chain algorithm for the aggregation estimators.  For the lasso, black (0\%) to white (100\%) indicates the the percentage of selection out of 100 replications of cross validation.
  Structured sparse aggregation is able to best recover the true sparsity pattern. Both sparse aggregation and the lasso suffer from false positives scattered throughout the lattice of candidate predictors.}
\label{twoDimOneRunFigure}
\end{figure}

\begin{figure}[t]
\centering
\includegraphics[width=5in]{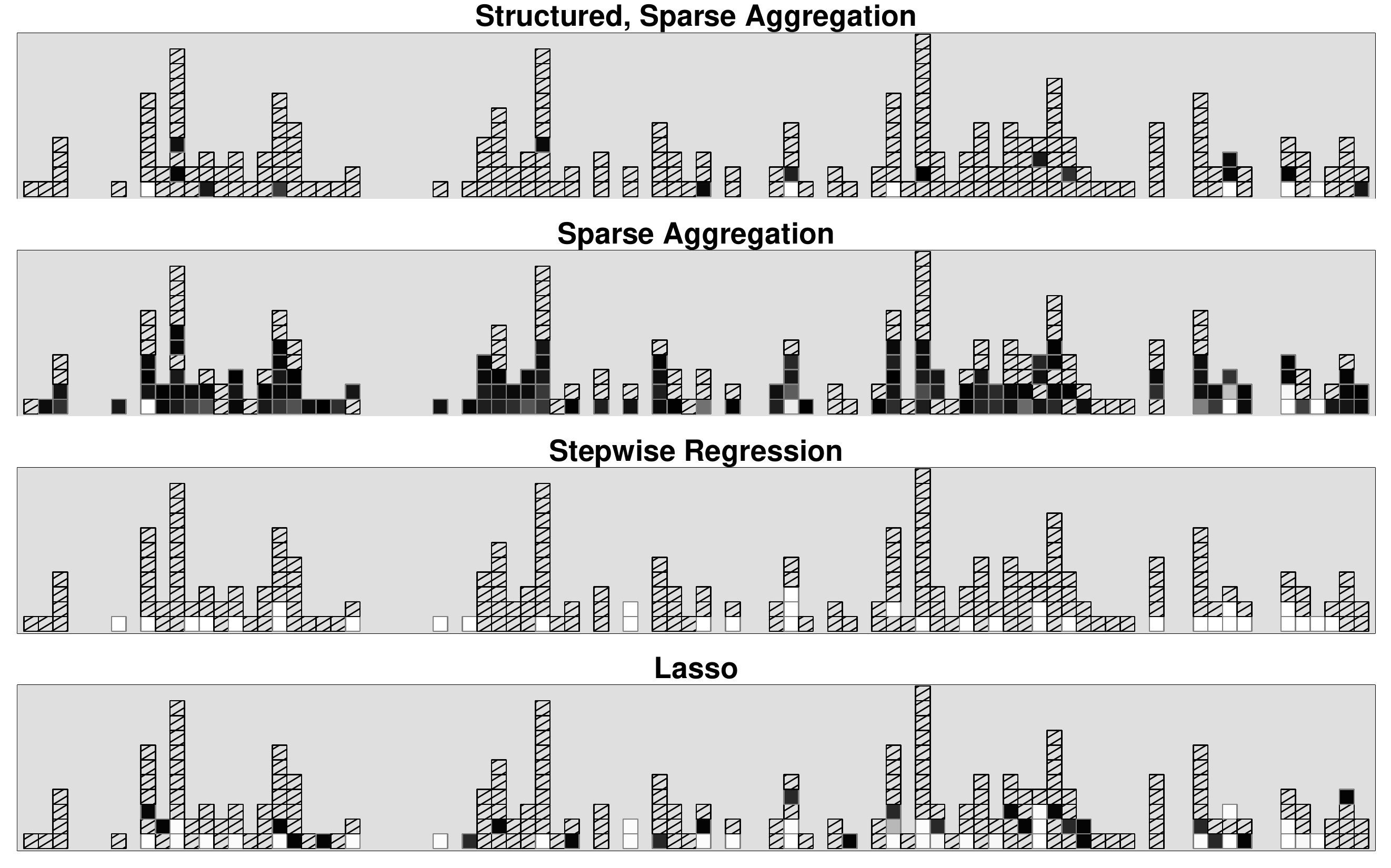}
\caption{Structured sparsity in an application: HIV drug resistance.  The panels give the selected sparsity patterns across HIV protein mutations for structured, sparse aggregation (top), sparse
  aggregation (second from top), stepwise regression
  (third from top), and the lasso (bottom). Each box represents a mutation covariate.  The horizontal
  axis represents location in the protein sequence. The locations ($1$
  to $99$) are arranged left to right as in the protein sequence. The vertical axis
  has no spatial meaning, each stack represents the number of mutations observed at that location in the protein sequence. Mutation predictors in adjacent bands are from adjacent locations in the protein sequence. Since proteins typically function via active sites, our structured model encourages clustered selection in the sequence.
For the aggregation methods, the color of the boxes indicates the
percentage of selection in the Markov Chain algorithm: Black (0\%) to white (100\%). If a mutation is never selected, it is gray and diagonally shaded. For the lasso, black (0\%) to white (100\%) indicates the the percentage of selection out of 100 replications of cross validation. For stepwise regression, we only report the selection a single instance of the algorithm.}
\label{hivFigure}
\end{figure}


\bibliographystyle{ECA_jasa}
\bibliography{refs}

\end{document}